\documentclass{llncs}
\usepackage{amssymb,enumitem,interval,xcolor,comment}
\usepackage{tuples,fixpoints,calculus,strut,galois,eqntabular,renumber,calculus}
\usepackage{stmaryrd}
\newcommand{\sqb}[1]{\llbracket#1\rrbracket}
\usepackage[hidelinks]{hyperref}
\begin{document}
\title{Nonstandard Axiomatic Semantics}
\author{Patrick Cousot}
\institute{Courant Institute, New York University}
\maketitle
\begin{abstract}
Similar to Skolem's nonstandard models of Peano's naturals, we show that axiomatic semantics based on Hoare logic has nonstandard models and so does not specify a unique, well-defined, and formal operational semantics of programming languages. We propose to enrich axiomatic semantics with additional proof obligations to solve this ambiguity problem. These proof obligations are always satisfied for standard trace models so that Hoare logic proofs are unchanged for these standard models.
\end{abstract}
\section{Introduction}
In addition to Peano's axiomatization \cite{Peano-1889}  of the naturals,
Dedekind's categoricity \cite{Dedekind-1888} eliminates Skolem's nonstandard models of arithmetic \cite{Skolem34} thus ensuring that the definition of the naturals is unique   and corresponds to the intended interpretation of the naturals, up to an isomorphism (section \ref{sec:NonstandardModelsArithmetics}). 

Since Hoare logic (section \ref{sec:Hoare:Logic}) is an abstraction (section \ref{sec:Hoare:Logic:Standard:Interpretation}) of a finite or infinite trace semantics indexed by the naturals \cite{DBLP:journals/tcs/Cousot02,DBLP:journals/pacmpl/Cousot24} (section \ref{sect:standard:trace:semantics}), an inevitable question is whether the logic is expressive enough to uniquely characterize this trace semantics, up to an isomorphism. This would show that Hoare logic is an operational characterization of program executions. On the contrary, we show that Lauer's
axiomatic semantics \cite{DBLP:phd/ethos/Lauer71,DBLP:journals/acta/HoareL74} based on Hoare logic \cite{DBLP:journals/cacm/Hoare69} suffers from an ambiguity problem similar to that of Peano arithmetics.
The rule for iteration (using an inductive invariant that must hold for all states of infinite executions) does not uniquely define the intended meaning of terminating and nonterminating loops (up to isomorphism) since many different infinite pasts and futures (section \ref{sec:nonstandard:Traces}) are possible without changing the logic theory (section \ref{sec:nonstandard:Trace:Semantics}). We study this phenomenon (sections \ref{sec:Nonstandard:Interpretation} and \ref{sec:Ambiguity:Axiomatic:Semantics}) and propose a remedy (section \ref{sect:logic:standardisation}) inspired by  termination proof methods (section \ref{sec:Termination:Proof:Methods}) extended so as to allow for non-termination. The additional requirements of the axiomatic semantics rule out nonstandard semantics and are always satisfied by the standard semantics (section \ref{sec:StandardHoareLogics}). It follows that Hoare logic proofs are unchanged under the hypothesis that Hoare logic  should be interpreted with respect to the standard semantics.

\section{Standard Trace Semantics of an Imperative Language}\label{sect:standard:trace:semantics}
We consider the standard trace semantics \cite{DBLP:conf/praapr/Wegner72} of a simple imperative language with iteration $\mathcal{L}$ (e.g.\ \cite[Ch.~7]{DBLP:books/mit/C2021}).  We let $\ell,\ell',\ldots \in\mathbb{L}$ be a set of labels/program points, $\texttt{x},\texttt{y},\ldots\in\mathbb{X}$ be a set of variables, $v, 0, \textit{true},\ldots\in\mathbb{V}$ be a set of values,  $\rho,\rho',\ldots\in\mathbb{E}\triangleq\mathbb{X}\rightarrow\mathbb{V}$ be a set of environments, and $\pair{\ell}{\rho},\ldots\in\Sigma\triangleq\mathbb{L}\times\mathbb{E}$ be a set of states. Executions $\sigma,\sigma',\sigma_b,\ldots\in \Sigma^{\ast\infty}\triangleq\Sigma^{\ast}\cup\Sigma^{\infty}$ are modeled by finite traces $\Sigma^{\ast}\triangleq\bigcup_{n\in\mathbb{N}}\:[0,n]\rightarrow\Sigma$ and infinite traces $\Sigma^{\infty}\triangleq\mathbb{N}\rightarrow\Sigma$. Trace concatenation is by juxtaposition of states and traces. A program semantics is a set of traces in $\wp(\Sigma^{\ast\infty})$.

The iteration statement \texttt{S} = \texttt{while${}^{\ell}\:$B\hskip0.5ex do$\:{}^{\ell'}$Sb${\:\:}^{\ell''}$} of $\mathcal{L}$ has unique labels $\ell\in\mathbb{L}$ at \texttt{S} and after the loop body \texttt{Sb}, $\ell'\in\mathbb{L}\setminus\{\ell\}$ at \texttt{Sb}, and $\ell''\in\mathbb{L}\setminus\{\ell,\ell'\}$ after the iteration \texttt{S}. 
The finite trace semantics $\sqb{\texttt{S}}^{\ast} \in\wp(\Sigma^{\ast})$ is a set of finite traces starting at $\ell$ and finishing at $\ell''$. The infinite trace semantics $\sqb{\texttt{S}}^{\infty}\in\wp(\Sigma^{\infty})$ is a set of infinite traces starting at $\ell$ that never terminate. These
semantics  may be empty (but not both). They are traditionally defined by structural induction (i.e.\ on the program syntax). Let $\sqb{\texttt{Sb}}^{\ast}$ (respectively $\sqb{\texttt{Sb}}^{\infty}$) be the standard finite (resp.\ infinite) trace semantics of the loop body \texttt{Sb}, already defined by structural induction. The finite executions in $\sqb{\texttt{Sb}}^{\ast}$ have the form $\pair{\ell'}{\rho}\sigma_b\pair{\ell}{\rho'}$ since they start at \texttt{Sb} that is $\ell'$ and terminate after \texttt{Sb} that is $\ell$. These standard trace semantics are defined using least $\lfp$ and greatest $\gfp$  fixpoints of $\subseteq$-increasing functions $B^{\ast}\sqb{\texttt{S}}$ and $F^{\infty}\sqb{\texttt{S}}$ on the complete lattice $\pair{\wp(\Sigma^{\ast\infty})}{\subseteq}$ \cite{Tarski1955}, as follows (in these formulas, $\ell,\ell',\ell''\in\mathbb{L}$ are labels of a statement \texttt{S}, not local dummy variables of the formulas. So we write $\{\ell\}$ instead of the more rigorous $\{x\mid x=\ell\}$ where  the dummy variable $x$ is local to the formula and $\ell$ must be defined globally to be a label of a statement \texttt{S}).
\begin{eqntabular}[fl]{rcl}
B^{\ast}\sqb{\texttt{S}}X&\triangleq&\{\pair{\ell}{\rho}\pair{\ell''}{\rho}\mid\neg\mathcal{B}\sqb{\texttt{B}}\rho\}\cup
\{\pair{\ell}{\rho}\pair{\ell'}{\rho}\sigma_b\pair{\ell}{\rho'}\sigma\mid 
\mathcal{B}\sqb{\texttt{B}}\rho\wedge{}\nonumber\\
&&\qquad\pair{\ell'}{\rho}\sigma_b\pair{\ell}{\rho'}\in \sqb{\texttt{Sb}}^{\ast}\wedge\pair{\ell}{\rho'}\sigma \in X\}\label{eq:def:Bast}\\
F^{\infty}\sqb{\texttt{S}}X&\triangleq&\{\pair{\ell}{\rho}\pair{\ell'}{\rho}\sigma_b\pair{\ell}{\rho'}\sigma\mid \mathcal{B}\sqb{\texttt{B}}\rho\wedge
\pair{\ell'}{\rho}\sigma_b\pair{\ell}{\rho'}\in \sqb{\texttt{Sb}}^{\ast}\wedge{}\label{eq:def:Foo}\\
&&{}
\phantom{\{}\pair{\ell}{\rho'}\sigma\in X\}{}\cup{}\{\pair{\ell}{\rho}\pair{\ell'}{\rho}\sigma_b\mid \mathcal{B}\sqb{\texttt{B}}\rho\wedge
\pair{\ell'}{\rho}\sigma_b\in \sqb{\texttt{Sb}}^{\infty}\}\nonumber\\
\sqb{\texttt{S}}^{\ast}&\triangleq&\Lfp{\subseteq}B^{\ast}\sqb{\texttt{S}}\colsep{=}\Gfp{\subseteq}B^{\ast}\sqb{\texttt{S}}
\quad~
\sqb{\texttt{S}}^{\infty}\colsep{\triangleq}\Gfp{\subseteq}F^{\infty}\sqb{\texttt{S}}
\quad~
\sqb{\texttt{S}}^{\ast\infty}\colsep{\triangleq}\rlap{$\sqb{\texttt{S}}^{\ast}\cup\sqb{\texttt{S}}^{\infty}$ .}\nonumber
\end{eqntabular}
(\ref{eq:def:Bast}) reflects the fact that \texttt{while B do Sb} is equivalent to
\texttt{if B then (S; while B do Sb)}. (\ref{eq:def:Foo}) reflects the fact that a loop does not terminate either because \texttt{B} is always true after zero or more terminating executions of its body or because after a number of iterations, the loop body does not terminate. The fixpoint equality $\Lfp{\subseteq}B^{\ast}\sqb{\texttt{S}}=\Gfp{\subseteq}B^{\ast}\sqb{\texttt{S}}$ follows from \cite[Theorem 11]{DBLP:journals/tcs/Cousot02}.

In examples with only one variable like \texttt{x:=0; while${}^{\ell}$ true do ${}^{\ell'}$ x:=x+1; ${}^{\ell''}$} we simplify the trace by abstracting states $\pair{\ell'}{\rho}$ by 
$\rho(\texttt{x})$ when $\ell'\in\{\ell,\ell''\}$ and ignoring the other states in the trace e.g. $0\:1\:2\:\:\ldots$.

By defining the computational ordering $X\sqsubseteq Y \triangleq ((X\cap\Sigma^{*})\subseteq(Y\cap\Sigma^{*}))\wedge ((X\cap\Sigma^{\infty})\supseteq(Y\cap\Sigma^{\infty}))$, $\pair{\wp(\Sigma^{\ast\infty})}{\sqsubseteq}$ is a complete lattice, $B^{\ast}\sqb{\texttt{S}}$ and $F^{\infty}\sqb{\texttt{S}}$ are $\sqsubseteq$-increasing, and we have $\sqb{\texttt{S}}^{\ast\infty}=\Lfp{\sqsubseteq}\LAMBDA{X}B^{\ast}\sqb{\texttt{S}}(X\cap\Sigma^{*})\cup F^{\infty}\sqb{\texttt{S}}(X\cap\Sigma^{\infty})$  \cite{DBLP:conf/popl/CousotC92}.

\section{Hoare Logic}\label{sec:Hoare:Logic}
Hoare logic \cite{DBLP:journals/cacm/Hoare69} uses triples $\{P\}\,S\,\{Q\}$ (originally written $P\,\{S\}\, Q$ in \cite{DBLP:journals/cacm/Hoare69}) to specify that any execution of statement $S$ from a state in $P$, which terminates (if ever), will terminate in a state satisfying $Q$. $P$ and $Q$ are assertions on (the values of)  program variables. Auxiliary variables not appearing in \texttt{S} can be introduced to record the initial value of variables in $P$ and express a relation between initial and final values in $Q$.

To avoid the inexpressivity problems of Hoare logic related to the choice of a logic to specify $P$ (or $Q$) \cite{DBLP:journals/siamcomp/Cook78,DBLP:journals/siamcomp/Cook81}, we assume (unless specified otherwise) that predicates are defined in extension as the set of all states satisfying $P$ (or $Q$).
Hoare defined the semantics of a simple imperative programming language by the following rules, which are literally extracted from \cite{DBLP:journals/cacm/Hoare69}, except that predicates are defined in extension as sets of environments $\rho$ mapping variables \texttt{x} to their value $\rho(\texttt{x})$ satisfying these predicates. We define the semantic assignment $\rho[\texttt{x}\leftarrow v]$ as $\rho[\texttt{x}\leftarrow v](\texttt{x})=v$ and $\rho[\texttt{x}\leftarrow v](\texttt{y})=\rho(\texttt{y})$ when $\texttt{y}\neq {x}$. We assume expressions \texttt{f} and $B$ are well-formed first order logic formulas without quantifiers (hence subject to nonstandard interpretations) where free variables are program variables or auxiliary variables.
\begin{quote}
\begin{itemize}
\item[\llap{``}\labelitemi] Axiom of Assignment: $\{P\}\,\texttt{x := f}\,\{\rho[\texttt{x}\leftarrow \sqb{\texttt{f}}(\rho)]\mid \rho\in P{}\}$ where $\sqb{\texttt{f}}(\rho)$ is the  value of expression $f$ when evaluated with values $\rho(\texttt{y})$ of the variables \texttt{y} of \texttt{f};
\item (Left and Right) Rules of Consequence:
\begin{itemize}[topsep=0pt,itemsep=0pt,parsep=0pt]
\item If $\{P\}\,Q\,\{R\}$ and $R\subseteq S$ then $\{P\}\,Q\,\{S\}$ (where logical implication $R\subseteq S$ is written $R\supset S$ in \cite{DBLP:journals/cacm/Hoare69});
\item If $\{P\}\,Q\,\{R\}$ and $S\subseteq P$ then $\{S\}\,Q\,\{R\}$;
\end{itemize}

\item Rule of Composition:\\
If $\{P\}\,Q_1\,\{R_1\}$ and $\{R_1\}\,Q_2\,\{R\}$ then $\{P\}\,Q_1\texttt{;}Q_2\,\{R\}$

\item Rule of Iteration:\\
If $\{P\cap B\}\,S\,\{P\}$ then $\{P\}\,\texttt{while}\,B\,\texttt{do}\,S\,\{\neg B \cap P\}$ where $B$ (respectively $\neg{B}$) stands for $\{\rho\mid{}$the evaluation of Boolean expression \texttt{B} for the values $\rho(\texttt{x})$ of the variables \texttt{x} appearing in \texttt{B} is true$\}$ (respectively false).''

\end{itemize}
\end{quote}
The derived union and intersection rules, which from $\forall i\in\Delta\mathrel{.}\{P_i\}\,\texttt{S}\,\{Q_i\}$ infer $\{\bigcup_{i\in\Delta}P_i\}\,\texttt{S}\,\{\bigcup_{i\in\Delta}Q_i\}$ and $\{\bigcap_{i\in\Delta}P_i\}\,\texttt{S}\,\{\bigcap_{i\in\Delta}Q_i\}$,  can be proved by induction on the structure of programs \cite{DBLP:conf/oopsla/CousotCLB12}.

\begin{remark}\label{rem:Hoare-on-oo-traces}Although Hoare logic is for partial correctness only, whose definition completely ignores infinite behaviors, something is nevertheless said about these infinite behaviors by the loop invariant when the precondition does not exclude infinite behaviors. For example, $P\triangleq\texttt{x>0}$ in $\{P\cap \texttt{x>0}\}\,\texttt{x:=x+1;}\,\{P\}$ for $\{P\}\,\texttt{while x>0}\hskip0.75ex\texttt{do}\,\texttt{x:=x+1;}\,\{\textit{false}\}$ provides the reachability information that infinite iterations all have \texttt{x} strictly positive.
\end{remark}

\section{The Standard Interpretation of Hoare Logic is an Abstraction of the Standard Trace Semantics}\label{sec:Hoare:Logic:Standard:Interpretation}
Define, for all statements $\texttt{S}\in\mathcal{L}$ where $\ell$ is at \texttt{S} and $\ell'$ is after \texttt{S}, the abstraction $\alpha_{\texttt{S}}$\ of a standard trace semantics $T$ into the set of all Hoare triples valid for that semantics as follows  \cite{DBLP:journals/pacmpl/Cousot24}.
\begin{eqntabular*}{rcl}
\alpha_{\texttt{S}}(T)&\triangleq\bigl\{\{P\}\texttt{S}\{Q\}\bigm| P,Q\in\wp(\mathbb{E})\wedge\forall \pair{\ell}{\rho}\sigma\pair{\ell'}{\rho'}\in T\cap\Sigma^{\ast}\mathrel{.}
\rho\in P\Rightarrow \rho'\in Q\bigr\}
\end{eqntabular*}
The abstraction $\alpha_{\texttt{S}}$ defines the standard interpretation of Hoare triples as partial correctness. $\alpha_{\texttt{S}}$\ is the lower adjoint of a Galois connection \cite{DBLP:journals/pacmpl/Cousot24}. Hoare logic theory $\mathcal{H}(\mathcal{L})$ for the language $\mathcal{L}$ is the set
\begin{eqntabular}{rcl}
\mathcal{H}(\mathcal{L})&\triangleq&\bigcup\bigl\{\alpha_{\texttt{S}}(\sqb{\texttt{S}}^{\ast\infty})\bigm|\texttt{S}\in\mathcal{L}\bigr\}
\label{def:standard:Hoare:logic:theory}
\end{eqntabular}
of all Hoare triples of all statements of the language $\mathcal{L}$ \cite{DBLP:journals/pacmpl/Cousot24}. Hoare proof rules follow from this definition by considering the abstraction $\alpha_{\texttt{S}}$\ of the structural fixpoint definition of the standard trace semantics $\sqb{\texttt{S}}^{\ast\infty}$ \cite{DBLP:conf/birthday/Cousot26}.

\section{Nonstandard Models of Arithmetic}\label{sec:NonstandardModelsArithmetics}
Peano \cite{Peano-1889}  defined the naturals $\mathbb{N}$ by $0\in\mathbb{N}$ and $\forall n\in\mathbb{N}\mathrel{.}\mathcal{S}n\in\mathbb{N}$ where $\mathcal{S}n=n+1$  is the successor function, together with $\forall n\in\mathbb{N}\mathrel{.}\mathcal{S}n\neq 0$ and injectivity $\forall n,m\in\mathbb{N}\mathrel{.}n+1=m+1\Rightarrow n=m$.
The intention was to define the naturals $\mathbb{N}$ as $0<1<2<3<\ldots$. Skolem \cite{Skolem34} showed that there are nonstandard models of Peano's definition such as the model $S$ =
$0<1<2<3<\ldots<\ldots-2<-1<0'<1'<2'\ldots$ which obviously satisfies Peano's axioms. Peano's arithmetic is not well-defined since, up to isomorphism, the axiomatization has more than one model.

Dedekind's idea \cite{Dedekind-1888} to ensure the uniqueness of the model (up to an isomorphism)  relies on the chain principle (or second-order induction) that is essentially De Morgan's mathematical induction (proofs by recurrence) \cite{DeMorgan1838}. From $P(0)$ and $\forall n\mathrel{.}P(n)\Rightarrow P(n+1)$ conclude $\forall n\in\mathbb{N}\mathrel{.}P(n)$ (which is also the contrapositive form of Fermat's earlier infinite descent \cite{Fermat1659-descente,DBLP:journals/igpl/Wirth04}, from $\forall n\mathrel{.}\neg P(n+1)\Rightarrow \neg P(n)$ and $P(0)$ conclude $\forall n\in\mathbb{N}\mathrel{.}P(n)$). Dedekind's additional  requirement is that the models of Peano's naturals are those for which proof by recurrence is valid (sound, correct). This eliminates the Skolem's model $S$ since for $P=\{0,1,2,\ldots\}$ (which is a definition  in extension of a predicate where $P(n)$ denotes $n\in P$) the hypotheses of the recurrence principle are satisfied  by $P$ while the conclusion $S\subseteq P$ is not. So the soundness of the principle of proof by recurrence  effectively limits Peano's definition to only those elements generated by iterating the  successor function $\mathcal{S}$ starting from $0$. 

Observe that $\mathbb{N}$ is the smallest model of Peano's naturals since $\mathbb{N}\setminus\{n\}$, $n\in\mathbb{N}$ does not satisfy Peano's definition.
This is Dedekind 's idea of defining $\mathbb{N}$ as the smallest subset of an infinite set such that there exists a transformation (the successor function $\mathcal{S}$)  that is injective and a base element $0$  not in the range of $\mathcal{S}$ \cite[\S71]{Dedekind-1888}. Otherwise stated, the naturals are the least fixpoint
\begin{eqntabular}{c}
\mathbb{N}=\Lfp{\subseteq}F\label{eq:def:nat:lfp}
\end{eqntabular}
of a transformer $F\in\wp(\mathbb{U})\rightarrow\wp(\mathbb{U})$ defined
as $F(X)\triangleq\{0\}\cup\{\mathcal{S}n\mid n\in X\}$  operating on the powerset of a universe $\mathbb{U}$ which can be chosen for example as finite sentences over an alphabet containing $0$, $1$, and $\mathcal{S}$ possibly $+$, $-$, $\times$, and any other required letters. Then the powerset $\pair{\wp(\mathbb{U})}{\subseteq}$ is a complete lattice and $F$ is $\subseteq$-increasing (monotone) so that
the $\subseteq$-least fixpoint $\Lfp{\subseteq}F$ of $F$ exists by Tarski's theorem \cite{Tarski1955} and is unique by definition of ``least''. This definition allows us to recover proof by recurrence thanks to Park's fixpoint induction \cite{Park1970-Fixpoint-Induction} (itself a direct consequence of Tarski's fixpoint theorem \cite{Tarski1955}) as follows ($P\in\wp(\mathbb{U})$)
\begin{calculus}
\formula{\mathbb{N}\subseteq P}\\
$\Leftrightarrow$
\formulaexplanation{\Lfp{\subseteq}F\subseteq P}{fixpoint definition $\mathbb{N}=\Lfp{\subseteq}F$}\\
$\Leftrightarrow$
\formulaexplanation{\exists I\in\wp(\mathbb{U})\mathrel{.} F(I)\subseteq I\wedge I\subseteq P}{fixpoint induction \cite{Park1970-Fixpoint-Induction}}\\
$\Leftrightarrow$
\formulaexplanation{\exists I\in\wp(\mathbb{U})\mathrel{.} \{0\}\cup\{Sn\mid n\in I\}\subseteq I\wedge I\subseteq P}{definition of $F$}\\
$\Leftrightarrow$
\formulaexplanation{\exists I\in\wp(\mathbb{U})\mathrel{.} 0\in I\wedge \forall n\in I\mathrel{.} Sn \in I \wedge I\subseteq P}{definition of $\cup$ and $\subseteq$}
\end{calculus}
This\ulstrut\  is exactly second-order induction also called recurrence principle (except that it is emphasized that $P$ might have to be strengthened to an inductive argument $I$ to be provable by recurrence, a situation which is so common in mathematics that it does not even need to be mentioned in the recurrence principle). We conclude that the fixpoint definition of the naturals $\mathbb{N}$ as ``the smallest set such that $0\in\mathbb{N}$ and $\forall n\in \mathbb{N}\mathrel{.} \mathcal{S}n \in \mathbb{N}$'' is equivalent to (and much simpler than) the mathematical definition requiring ``$0\in\mathbb{N}$, $\forall n\in \mathbb{N}\mathrel{.} \mathcal{S}n \in \mathbb{N}$, and the recurrence principle is valid in $\mathbb{N}$''.

Notice that there are different possible universes but the defined naturals in each case are isomorphic, this is Dedekind's categoricity theorem. Examples are Zermelo naturals $0=\{\}$ and $n+1=\{n\}$ in his axiom of infinity \cite[Axiom VII]{Zermelo-sets-1908} or Von Neumann naturals $0=\{\}$ and $n+1=n\cup \{n\}$ as part of the ordinals \cite{VonNeumann-ordinals-1923}.

\section{Nonstandard Traces of an Imperative Language}\label{sec:nonstandard:Traces}
Let $\mathbb{N}^{+\delta}$ be the standard naturals $\mathbb{N}$ indexed by $0$ followed by $\delta\in\mathbb{N}$ copies of $\mathbb{Z}$ indexed by $i\in[1,\delta]$ totally ordered by the lexicographic ordering on the index and the integers. Informally 
$\mathbb{N}^{+\delta}=0_0\ 1_0\ 2_0\ \ldots\ \ldots -2_1\ -1_1\ 0_1\ 1_1\ 2_1\ldots\quad\ldots\quad\ldots \mbox{$-2_{\delta}$}\allowbreak\ \allowbreak-1_{\delta}\ 0_{\delta}\ 1_{\delta}\ 2_{\delta}\ldots$ 
So $\mathbb{N}^{+0}=\{n_0\mid n\in\mathbb{N}\}$ is isomorphic to $\mathbb{N}$ while $\mathbb{N}^{+1}=\{n_0\mid n\in\mathbb{N}\}\cup \{z_1\mid z\in\mathbb{Z}\}$ is isomorphic to the Skolem model $S$ in section \ref{sec:NonstandardModelsArithmetics}. 

Assuming that states $\pair{\ell}{\rho}$ are abstracted to $\rho(\texttt{x}$) and others ignored, the $\mathbb{N}^{+\delta}$, $\delta\in\mathbb{N}$ are trace models of the program \texttt{while true do x=x+1;} where $\texttt{x}\in\mathbb{N}^{+\delta}$ that starts and never terminates. Another example suggested by a reviewer is ``\texttt{y:=0; while (y==0) do x:=x+1;}, and then
have a trace with Skolem’s nonstandard natural numbers as domain. After each of
the standard points 0, 1, ..., n, we will have y = 0; but on the nonstandard
points -3, -2, -1, 0', 1', 2', ... one could choose y to be an arbitrary
constant > 0.''

Define $\mathbb{N}^{-\delta}\triangleq\{-m_{-i}\mid m_i\in \mathbb{N}^{+\delta}\}$. 
Informally 
$\mathbb{N}^{{-}\delta}=
\ldots {-}2_{{-}\delta}\ {-}1_{{-}\delta}\ 0_{{-}\delta}\ 1_{{-}\delta}\allowbreak\ \allowbreak2_{{-}\delta}\ldots
\quad\ldots\quad
\ldots {-}2_{{-}1}\ {-}1_{{-}1}\ 0_{{-}1}\ 1_{{-}1}\ 2_{{-}1}\ldots\ 
\ldots\ {{-}}2_{0}\ {{-}}1_{0}\ 0_{0}$ .The $\mathbb{N}^{{-}\delta}$, $\delta\in\mathbb{N}$ are trace models of program \texttt{while x!=0 do x=x+1;} where $x\in\mathbb{N}^{{-}\delta}$. 
In addition define $\mathbb{N}^{{-}\delta+\delta'}\triangleq \mathbb{N}^{+\delta}\cup\mathbb{N}^{-\delta'}$ and $\mathbb{N}^{\pm\delta}\triangleq\mathbb{N}^{-\delta+\delta}$. For example $\mathbb{N}^{\pm1}=\{ \ldots\ -2_{-1}\ -1_{-1}\ 0_{-1}\ 1_{-1}\ 2_{-1}\ \ldots\mbox{$-2_0$}$ $\mbox{$-1_0$}\ 0_0\ 1_0\ 2_0\ \ldots\mbox{$-2_1$}\ -1_1\ 0_1\ 1_1\ 2_1\ \ldots\}$.

We define nonstandard infinite traces as $\Sigma^{\widetilde{\infty}}\triangleq\Sigma^{\ast} \cup
(\mathbb{N}^{-\delta}\rightarrow\Sigma)\cup
(\mathbb{N}^{+\delta}\rightarrow\Sigma)\cup
(\mathbb{N}^{\pm\delta}\rightarrow\Sigma)$ as finite traces in $\Sigma^{\ast}$ plus nonstandard backward traces in $\mathbb{N}^{-\delta}\rightarrow\Sigma$ that never start but terminate, nonstandard forward traces in $\mathbb{N}^{+\delta}\rightarrow\Sigma$ that start but never terminate, and nonstandard backward/forward traces in $\mathbb{N}^{\pm\delta}\rightarrow\Sigma$ that neither start nor terminate in the transfinite. 

The parts of a trace indexed by $\mathbb{N}^{+0}$ or $\mathbb{N}^{-0}$ are called the standard parts of the trace and those indexed by $\mathbb{N}^{-\delta}$ or $\mathbb{N}^{+\delta}$, $\delta>0$ are called the nonstandard parts.

\section{Nonstandard Trace Semantics of an Imperative Language}\label{sec:nonstandard:Trace:Semantics}

The nonstandard trace semantics of a simple imperative language is an element of $\wp(\Sigma^{\widetilde{\infty}})$. It can be given fixpoint definitions.
Consider for example the traces that terminate but never start for the iteration statement \texttt{S} = \texttt{while${}^{\ell}\:$B\hskip0.5ex do$\:{}^{\ell'}$Sb~${}^{\ell''}$}. In the following formulas, we assume that $\ell$, $\ell'$, and $\ell''$ are the labels of \texttt{S} (i.e.\ not variables local to the formula).
Let $\textsf{Seed}^{-\delta}\sqb{\texttt{S}}X$ be the subset of nonstandard backward traces of $X$ such that the standard and nonstandard parts all have at least one state
of the form $\pair{\ell}{\rho}$ such that $\mathcal{B}\sqb{\texttt{B}}\rho$ holds.
Define \begin{eqntabular}[fl]{rcl}
B^{-\delta}\sqb{\texttt{S}}X&\triangleq&\{\sigma_1\pair{\ell}{\rho}\pair{\ell'}{\rho}\sigma_2\pair{\ell}{\rho'}\pair{\ell''}{\rho'}\in\textsf{Seed}^{-\delta}\sqb{\texttt{S}}X\mid {}
\label{eq:def:B-delta}\stepcounter{equation}\renumber{(\ref{eq:def:B-delta}.a)}\\
&&\quad\mathcal{B}\sqb{\texttt{B}}\rho\wedge\pair{\ell'}{\rho}\sigma_2\pair{\ell}{\rho'}\in\sqb{\texttt{Sb}}^{\ast}{}\wedge{}\neg\mathcal{B}\sqb{\texttt{B}}\rho'\}{}\cup{}\nonumber\\
&&\{\sigma_1\pair{\ell}{\rho'}\pair{\ell'}{\rho'}\sigma'_1\pair{\ell}{\rho}\pair{\ell'}{\rho}\sigma'_2\pair{\ell}{\rho''}\sigma_2\in\textsf{Seed}^{-\delta}\sqb{\texttt{S}}X\mid{}\renumber{(\ref{eq:def:B-delta}.b)}\\
&&\quad\mathcal{B}\sqb{\texttt{B}}\rho'\wedge{}
\pair{\ell'}{\rho'}\sigma'_1\pair{\ell}{\rho}\in\sqb{\texttt{Sb}}^{\ast}{}\wedge{}\
\mathcal{B}\sqb{\texttt{B}}\rho\wedge{}\nonumber\\
&&\qquad\pair{\ell'}{\rho}\sigma'_2\pair{\ell}{\rho''}\in\sqb{\texttt{Sb}}^{\ast}{}\wedge{}
\mathcal{B}\sqb{\texttt{B}}\rho''\}
\nonumber\\
\sqb{\texttt{S}}^{-\delta}&\triangleq&\Gfp{\subseteq}B^{-\delta}\sqb{\texttt{S}}\nonumber
\end{eqntabular}
The $\gfp$ is the intersection of the transfinite fixpoint iterations starting from all possible nonstandard backward traces $\mathbb{N}^{-\delta}\rightarrow\Sigma$. If some part of a trace has no seed state $\pair{\ell}{\rho}$ with $\mathcal{B}\sqb{\texttt{B}}\rho$ true, the trace is eliminated from the next fixpoint iterate so, at the limit by intersection of the fixpoint iterates, from the $\gfp$. So after one fixpoint iteration all parts of all nonstandard backward traces have a seed state $\pair{\ell}{\rho}$ where \texttt{B} holds. Then in the next fixpoint iteration, at least one seed state $\pair{\ell}{\rho}$ in the trace is required 
\begin{itemize}[topsep=0pt,itemsep=0pt]
\item[(\ref{eq:def:B-delta}.a)] either to be the last in the standard part before termination, in which case it is eliminated if it is not preceded by an execution of the loop body \texttt{Sb} from $\ell$ to $\ell$ with \texttt{B} finally false;
\item[(\ref{eq:def:B-delta}.b)] or, if not the last in the standard part to be preceded and followed by executions of the body \texttt{Sb} from $\ell$ to $\ell$ with \texttt{B} true.
\end{itemize}
It follows that each iteration ensures that the number of iterations around the seed is strictly increased. At the limit, the $\gfp$ which is the intersection of these transfinite fixpoint iterations must have all its standard and nonstandard parts built out of traces of a true test \texttt{B} followed by an execution of the loop body (\ref{eq:def:B-delta}.b), except for the standard part which must end with \texttt{B} false (\ref{eq:def:B-delta}.a). The difference with the standard semantics is that $B^{\ast}\sqb{\texttt{S}}$ in (\ref{eq:def:Bast}) requires successive loop iterations backward starting from the seed at the end of the trace while $B^{-\delta}\sqb{\texttt{S}}$ requires loop iterations before and after the seed that can be anywhere. For the standard part this will ultimately reach the end thanks to (\ref{eq:def:B-delta}.b) at which time termination is required by (\ref{eq:def:B-delta}.a). 

The nonstandard forward trace semantics is defined similarly as $\sqb{\texttt{S}}^{+\delta}\triangleq\Gfp{\subseteq}F^{+\delta}\sqb{\texttt{S}}$ where $F^{+\delta}\sqb{\texttt{S}}$ generalizes  $F^{\infty}\sqb{\texttt{S}}$ to  nonstandard traces. 

\begin{eqntabular}[fl]{l}
F^{+\delta}\sqb{\texttt{S}}X\colsep{\triangleq}\{\sigma_1\pair{\ell}{\rho}\sigma_2\in \textsf{Seed}^{+\delta}\sqb{\texttt{S}}X\mid \mathcal{B}\sqb{\texttt{B}}\rho{}\wedge{}\label{eq:def:F+delta}\stepcounter{equation}\nonumber\\
\quad
\bigl((\sigma_1=\epsilon)\vee(\exists\sigma'_1\pair{\ell}{\rho'}\pair{\ell'}{\rho'}\sigma''_1=\sigma_1\mathrel{.}\mathcal{B}\sqb{\texttt{B}}\rho'\wedge\pair{\ell'}{\rho'}\sigma''_1\pair{\ell}{\rho}\in\sqb{\texttt{Sb}}^{\ast})\bigr)\renumber{(\ref{eq:def:F+delta}.a)}\\
\quad\llap{${}\wedge{}$}\bigl((\exists\pair{\ell'}{\rho}\sigma'_2\pair{\ell}{\rho''}\sigma''_2=\sigma_2\mathrel{.}\pair{\ell'}{\rho}\sigma'_2\pair{\ell}{\rho''}\in\sqb{\texttt{Sb}}^{\ast}{}\wedge\mathcal{B}\sqb{\texttt{B}}\rho''){}\vee{}\renumber{(\ref{eq:def:F+delta}.b)}\\
\qquad(\exists\pair{\ell'}{\rho}\sigma_b=\sigma_2\mathrel{.} \mathcal{B}\sqb{\texttt{B}}\rho\wedge
\pair{\ell'}{\rho}\sigma_b\in \sqb{\texttt{Sb}}^{\infty}
)\bigr)
\}\renumber{(\ref{eq:def:F+delta}.c)}
\end{eqntabular}
Let $X$ be an iterate of $F^{+\delta}\sqb{\texttt{S}}$ (initially all nonstandard forward traces  $\mathbb{N}^{+\delta}\rightarrow\Sigma$). In the next fixpoint iterate 
$F^{+\delta}\sqb{\texttt{S}}X$ appear only the traces $\sigma_1\pair{\ell}{\rho}\sigma_2$ which have at least one seed state in all their parts and at least one of these seeds $\pair{\ell}{\rho}$ is, by (\ref{eq:def:F+delta}.a) either the first state in the trace or else preceded by a terminating  execution of the loop body \texttt{Sb}  starting with \texttt{B} true and, by (\ref{eq:def:F+delta}.b), followed by a terminating execution of the loop body \texttt{Sb}  ending with \texttt{B} true, or by (\ref{eq:def:F+delta}.c), followed by a nonterminating execution of the loop body \texttt{Sb} . (\ref{eq:def:F+delta}.a) and (\ref{eq:def:F+delta}.b) guarantee the existence of other seeds which at the next fixpoint iterate will be required to satisfy the same constraints. Passing to the limit by an intersection, after possibly transfinite fixpoint iterations, only infinite nonstandard loop traces will remain. In particular, it follows that $\sqb{\texttt{S}}^{+0}\triangleq\Gfp{\subseteq}F^{+0}\sqb{\texttt{S}}=\Gfp{\subseteq}F^{\infty}\sqb{\texttt{S}}\triangleq\sqb{\texttt{S}}^{\infty}$ in the standard case.

We also define
\begin{eqntabular*}{rcl}
\sqb{\texttt{S}}^{-\delta+\delta'}&\triangleq&\{\sigma\pair{\ell}{\rho}\sigma'\mid
\sigma\pair{\ell}{\rho}\in\Gfp{\subseteq}B^{-\delta}\sqb{\texttt{S}}\wedge
\pair{\ell}{\rho}\sigma'\in\Gfp{\subseteq}F^{+\delta'}\sqb{\texttt{S}}\}
\end{eqntabular*}
which has both  backward and forward infinite nonstandard traces.

The traditional trace semantics is $\sqb{\texttt{S}}^{\ast\infty}$ $\triangleq$
 $\sqb{\texttt{S}}^{\ast}\cup\sqb{\texttt{S}}^{\infty}$ =
$\sqb{\texttt{S}}^{\ast}\cup \sqb{\texttt{S}}^{+0}$ where $\sqb{\texttt{S}}^{\ast}$ is the classic finite semantics and $\sqb{\texttt{S}}^{\infty}=\sqb{\texttt{S}}^{+0}$ the standard infinite traces.

There are several possible choices for the nonstandard semantics, for example, $\sqb{\texttt{S}}^{\ast}\cup\sqb{\texttt{S}}^{-\delta}$ ignoring nontermination as in natural semantics \cite{DBLP:conf/stacs/Kahn87} (but keeping the traces that terminate but never start), $\sqb{\texttt{S}}^{\ast}\cup\sqb{\texttt{S}}^{-\delta}\cup\sqb{\texttt{S}}^{+\delta}$ or $\sqb{\texttt{S}}^{\widetilde{\infty}}$ $\triangleq$ $\sqb{\texttt{S}}^{\ast}\cup\sqb{\texttt{S}}^{-\delta}\cup\sqb{\texttt{S}}^{+\delta}\cup\sqb{\texttt{S}}^{-\delta+\delta}$. Then $\sqb{\texttt{while true do skip}}^{\widetilde{\infty}}=\sqb{\texttt{while true do skip}}^{+\delta}\cup\sqb{\texttt{while true do skip}}^{-\delta+\delta}$ since it either starts and never terminates or never starts and never terminates. 

\section{The Nonstandard Interpretation of the Axiomatic Semantics}\label{sec:Nonstandard:Interpretation}

Define, for all statements $\texttt{S}\in\mathcal{L}$, the abstraction of a nonstandard trace semantics
$T$ into the set of all Hoare triples valid for that semantics (again $\ell=\mathsf{at}(\texttt{S})$ and $\ell''=\mathsf{after}(\texttt{S})$ are global variables in the formula denoting program points).
\begin{eqntabular}[fl]{@{\qquad}rcl}
\widetilde{\alpha}_{\texttt{S}}(T)&\triangleq&\bigl\{\{P\}\texttt{S}\{Q\}\bigm| P,Q\in\wp(\mathbb{E})\wedge{}\nonumber\\
&&\renumber{$\begin{array}[t]{l}((\forall \pair{\ell}{\rho}\sigma\pair{\ell''}{\rho''}\in T\cap\mathbb{N}^{\ast}\rightarrow\Sigma\mathrel{.}
\rho\in P\Rightarrow \rho''\in Q{})\\
{}\vee{}(\forall \sigma\pair{\ell}{\rho}\sigma'\pair{\ell''}{\rho''}\in T\cap\mathbb{N}^{-\delta}\rightarrow\Sigma\mathrel{.}{}\\
\qquad\quad(\pair{\ell}{\rho}\sigma'\pair{\ell''}{\rho''}\in(\mathbb{N}^{-0}\rightarrow\Sigma)\wedge\rho\in P)\Rightarrow \rho''\in Q{}))\bigr\}
\end{array}$}
\end{eqntabular}
This is the standard requirement for finite traces. For backward nonstandard traces  that never start but terminate, their nonstandard part is ignored. If their standard part has a state $\pair{\ell}{\rho}$ satisfying $P$ then the final state must satisfy the postcondition $Q$. 
Notice that if auxiliary variables record the values of variables at $\ell$ in $P$ when $\ell$ is reached at any past loop iteration and these auxiliary variables appear in $Q$ then $Q$  must hold between these values of the auxiliary variables when $\ell$ was reached  and the final values of the program variables at $\ell''$.

The nonstandard Hoare logic theory $\widetilde{\mathcal{H}}(\mathcal{L})$ for  the language $\mathcal{L}$ is the set
\begin{eqntabular}{rcl}
\widetilde{\mathcal{H}}(\mathcal{L})&\triangleq&\bigcup\bigl\{\widetilde{\alpha}_{\texttt{S}}(\sqb{\texttt{S}}^{\widetilde{\infty}})\bigm|\texttt{S}\in\mathcal{L}\bigr\}\label{def:nonstandard:Hoare:logic:theory}
\end{eqntabular}

\section{Ambiguity of the Axiomatic Semantics}\label{sec:Ambiguity:Axiomatic:Semantics}
The standard $\mathcal{H}(\mathcal{L})$ and nonstandard $\widetilde{\mathcal{H}}(\mathcal{L})$ theories are the same.
\begin{theorem}\label{th:same:theories} $\widetilde{\mathcal{H}}(\mathcal{L})=\mathcal{H}(\mathcal{L})$
\end{theorem}
\begin{proof}Obviously ${\mathcal{H}}(\mathcal{L})\subseteq\widetilde{\mathcal{H}}(\mathcal{L})$. The converse follows from the fact that if $\sigma\pair{\ell}{\rho}\sigma'\pair{\ell''}{\rho''}\in \sqb{\texttt{S}}^{\widetilde{\infty}}\cap\mathbb{N}^{-\delta}\rightarrow\Sigma$ and $\pair{\ell}{\rho}\sigma'\pair{\ell''}{\rho''}\in(\mathbb{N}^{-0}\rightarrow
\Sigma)$ then $\pair{\ell}{\rho}\sigma'\pair{\ell''}{\rho''}\in\sqb{\texttt{S}}^{\ast}$ together with $\rho\in P$ implies $\rho''\in Q$ and therefore $\{P\}\texttt{S}\{Q\}$ holds by definition of ${\alpha}_{\texttt{S}}(\sqb{\texttt{S}}^{\ast})$.
\qed\end{proof}
Since the theories are the same by theorem \ref{th:same:theories}, their proof rules are the same by structural induction on the language $\mathcal{L}$ and Aczel's  correspondence between theories and proof rules \cite{Aczel:1977:inductive-definitions} for iteration. It follows that the popular statement that ``Axiomatic semantics is a formal approach in computer science that defines a program's meaning by its effects on logical assertions (pre- and post-conditions), rather than how it executes'' is wrong since the defined meaning is not unique up to isomorphism, unlike trace semantics do. This is the same problem as for Peano's original definition of the naturals.

\begin{example}
Consider the program \texttt{while x<0 do x:=x+1;} on integers whose standard traces in $\mathbb{N}^{\ast\infty}\rightarrow\Sigma=\mathbb{N}{\ast}\rightarrow\Sigma\cup\mathbb{N}^{+0}\rightarrow\Sigma$ are isomorphic to $\{-n\ -n+1\ \ldots\ 0\mid n\in\mathbb{N}\}$ and $\{n\mid n\in\mathbb{N}\}$ since the program executions starting with a negative value $-n$ of variable \texttt{x} will iteratively increment it by one until reaching $0$ and then exiting the loop or else the program executions will start with a value $n$ of \texttt{x} which is greater than or equal to $0$, in which case the loop is not entered and immediately exited. 

The method to show that Hoare logic includes any of these traces is for each one to define $P$ to be the set of states along that trace and
to prove that $P$ is a loop inductive invariant according to Hoare's rule of iteration. 

For example, consider the trace $-n\ -n+1\ \ldots\ 0$ for a given $n\in\mathbb{N}$ where $\rho\in\Sigma$ such that $\rho(\texttt{x})=k$ is encoded simply by $k$.
Let $P=\{\rho\mid \rho(\texttt{x})\in\interval{-n}{0}\}$ be the set of states along that trace. We have 
$P\cap\{\rho\mid \rho(\texttt{x})<0\}$ = $\{\rho\in P\mid \rho(\texttt{x})<0\}$ = $\{\rho\mid \rho(\texttt{x})\in\interval{-n}{-1}\}$ and \{$\{\rho\mid \rho(\texttt{x})\in\interval{-n}{-1}\}$\}\texttt{x:=x+1;}\{$\{\rho\mid \rho(\texttt{x})\in\interval{-n+1}{0}\}$\} by the assignment rule. Moreover $\{\rho\mid \rho(\texttt{x})\in\interval{-n+1}{0}\}\subseteq P$ so that $\{P\cap\{\rho\mid \rho(\texttt{x})<0\}\}\,\texttt{x:=x+1;}\,\{P\}$ holds by the consequence rule. By the iteration rule, $\{P\}\,\texttt{while}\:\texttt{x<0}\:\texttt{do}\,:\texttt{x:=x+1;}\,\{\neg(x<0) \cap P\}$ proving that the trace $-n\ -n+1\ \ldots\ 0$ is feasible for that program. Similarly for any $n\in\mathbb{N}$ we can prove \{\{n\}\}\texttt{while x<0 do x:=x+1;}\{\{n\}\} (considering only the value $n$ of \texttt{X}), proving the feasibility of the traces $\{n\mid n\in\mathbb{N}\}$.

Now consider the traces in $\mathbb{N}^{-0}\rightarrow\Sigma$ so that the program semantics now include the trace $\ldots\ {-2}\ {-1}\ 0$ that never starts and stops at $0$. The invariant is now $P=\{\rho\mid \rho(\texttt{x})\in\mathbb{Z}\wedge \rho(\texttt{x})\leqslant 0\}$ which also satisfies Hoare's verification conditions for iteration so that this trace is also proved to be feasible. 

It follows that program \texttt{while x<0 do x:=x+1;} has the standard semantics $\{-n\ -n+1\ \ldots\ 0\mid n\in\mathbb{N}\}\cup\{n\mid n\in\mathbb{N}\}$ in $\mathbb{N}^{\ast\infty}\rightarrow\Sigma$ and the nonstandard semantics $\{\ldots\ {-2}\ {-1}\ 0\}\cup\{-n\ -n+1\ \ldots 0\mid\ n\in\mathbb{N}\}\cup\{n\mid n\in\mathbb{N}\}$ in $\mathbb{N}^{-0}\rightarrow\Sigma$. This shows that for some programs,  Hoare logic cannot distinguish between the standard trace semantics on $\mathbb{N}^{\ast\infty}\rightarrow\Sigma$ and the nonstandard one on $\mathbb{N}^{\ast}\rightarrow\Sigma\cup\mathbb{N}^{-0}\rightarrow\Sigma$.

For the converse question, fix the semantics to be either on the standard $\Sigma^{\ast\infty}$ or on the nonstandard $\Sigma^{\ast}\cup\Sigma^{-0}$.
The question is whether there is a triple $\{P\}\,\texttt{S}\,\{Q\}$ allowing us to know which semantics is considered. This cannot be on finite executions and must be on infinite or transfinite executions, the only ones making a difference. This difference can only appear in the loop invariant in $\wp(\mathbb{E})$ for infinite trace. If this loop invariant is expressed in first-order logic then, by the compactness theorem of first-order logic, it is true for $\Sigma$ isomorphic to $\mathbb{N}$ if and only if it is true for $\Sigma$ isomorphic to nonstandard naturals  so no distinction is possible in this case. More generally, if the invariant holds for a transfinite trace, it must hold for its standard part, hence for the standard semantics (in which case it is an over-approximation, possibly referring to nonstandard naturals, up to an isomorphism with $\Sigma$). 
\end{example}
\begin{example}
Consider now the program \texttt{while true do x:=x+1;} which traces in $\Sigma^{\infty}=\Sigma^{+0}$ are $\{n\ n+1\ n+2\ \ldots \mid n\in\mathbb{Z}\}$ and traces on $\Sigma^{-0}\cup\Sigma^{+0}$ also contain the trace $\ldots\ -2\ -1\ 0\ 1\ 2\ \ldots$. Again for all of these traces the set of their states is a valid invariant in Hoare logic so that Hoare logic cannot distinguish
between the standard semantics on $\Sigma^{\ast\infty}$ and the nonstandard semantics on $\Sigma^{\ast-0+0}$.
\end{example}

\section{A Reminder on Termination Proof Methods}\label{sec:Termination:Proof:Methods}

Turing/Floyd termination proof method \cite{Turing49-program-proof,Floyd67-1} consists in exhibiting a variant function mapping program variables into a well-founded set and showing that its value strictly decreases on each loop iteration (or after a finite number of iterations).

Knuth's termination proof method \cite[Section 3a]{DBLP:journals/csur/Knuth74} consists in adding an integer counter to each loop which is initialized, e.g.\ to $0$, before entering the loop and incremented by a strictly positive value, e.g.\ by $1$,  within the loop body (the counter is assumed to be a mathematical integer so the incrementation cannot overflow).
If, thanks to the invariants statically relating the counter to the other variables in the loop, the counter value can be proved to be upper-bounded by $M\in\mathbb{N}$ then termination in finitely many steps is proved. The method is sound but incomplete, e.g.\ for unbounded non-determinism requiring transfinite ordinals \cite{CousotR85-1}.

This method was used in \cite[section 8.3]{Halbwachs79-these} to prove termination of programs by linear inequalities/polyhedral static analysis \cite{DBLP:conf/popl/CousotH78}. 

In the same way as Manna and Pnueli extended Hoare logic to total correctness, i.e.\ partial correctness with termination \cite{DBLP:journals/acta/MannaP74,DBLP:journals/jacm/AptP86}, it is easy to extend Hoare logic with Knuth's termination proof method \cite[section 2]{DBLP:journals/acta/KatzM75} ($t$, $\alpha$, and $M$ are auxiliary variables in $\mathbb{N}$ that do not appear in $P$, \texttt{B}, and \texttt{Sb}).

\begin{eqntabular}{c}
\frac{~P\wedge\texttt{B}\Rightarrow{t}<M,\quad \{P\wedge \texttt{B}\wedge t=\alpha\}\,\texttt{Sb}\,\{P\wedge t=\alpha+1\}~}{\{P\wedge t=0\}\,\texttt{while B do Sb}\,\{P\wedge\neg\texttt{B}\}}
\label{eq:Knuth:rule}
\end{eqntabular}

\section{Standardisation of the Axiomatic Semantics}\label{sect:logic:standardisation}

We now extend the axiomatic semantics so as to exclude nonstandard interpretations and show in section \ref{sec:StandardHoareLogics} that this does not change Hoare's proof method under the hypothesis that the logic is interpreted in the standard trace semantics of section \ref{sect:standard:trace:semantics}. In fact this hypothesis is implicit in most uses of Hoare logic.

For naturals, considering the least model eliminates nonstandard models. A similar idea would be to consider the strongest loop invariant. But this does not help, e.g.\ \texttt{x:=0; while true do skip;} has strongest invariant $\texttt{x}=0$ which does not exclude forward nonstandard traces. Anyway, Hoare logic cannot check that a given invariant is the strongest possible one.

\subsection{Preventing Nonstandard Infinite Forward Traces}\label{sec:no:nonstandard:forward:traces}

To prevent the nonstandard forward traces of section \ref{sec:nonstandard:Traces}, we extend Knuth's method by taking the counter to be an ordinal which is initialized to $0$ before the loop, and the incrementation within the loop is the ordinal addition of 1 to the least upper bound of all previous values of the counter. It follows that for a standard infinite forward iteration the counter is bounded by $\omega$ where $\omega$ is the smallest infinite limit ordinal. However, for a nonstandard infinite forward iteration the counter should be incremented by one at each step in the nonstandard part which is impossible since there is no order isomorphism between integers and ordinals. So adding this further proof obligation in Hoare's logic for loops will guarantee that the counter for that loop is upper bounded by $\omega$ which will definitely exclude nonstandard forward infinite traces.

Knuth's auxiliary counter is equivalent to the existence of a variant function $K$ recording the value of the counter over time. Let $\mathbb{O}$ be the class of Von Neumann ordinals \cite{VonNeumann-ordinals-1923}, $\oplus$ be ordinal addition,  and $\omega$ be the first infinite limit ordinal. The following requirement limits loop iterations to the standard first  $\omega$ ones.
\begin{eqntabular}{r}
\exists K\in \mathbb{N}^{+\delta}\rightarrow\mathbb{O}\ \textrm{increasing}\ \mathrel{.}\forall n\in\mathbb{N}^{+\delta}\mathrel{.} P\wedge \texttt{B}\Rightarrow K(n)<\omega{}\wedge{}\qquad\label{eq:no:nonstandard:forward:traces}\\
\forall\alpha\in\mathbb{O}\mathrel{.}\{P\wedge\texttt{B}\wedge K(n)=\alpha\}\,\texttt{Sb}\,\{P\wedge K(n+1)=\alpha\oplus 1\}\nonumber
\end{eqntabular}
In (\ref{eq:Knuth:rule}), $t$ is the number of iterations so far in the loop. It is replaced in 
(\ref{eq:no:nonstandard:forward:traces}) by $K(n)$ which maps the position $n\in\mathbb{N}^{+\delta}$ in the trace to the number of iterations so far in the loop.
So $n$ should be initialized to zero in $\{P\wedge n=0\}\,\texttt{while B do Sb}\,\{P\wedge\neg\texttt{B}\}$ and incremented after each test or assignment e.g. $\{P\}\,\texttt{x := f}\,\{\rho[\texttt{x}\leftarrow \sqb{\texttt{f}}(\rho),n\leftarrow \rho(n)+1]\mid \rho\in P{}\}$ where the incrementation of $\rho(n)$ on $\mathbb{N}^{+\delta}$ is nonstandard.

$K$ must be increasing (monotone) and there exists no order-isomorphism between nonstandard integer parts and 
ordinals since ordinals are well-founded and integers admit infinite strictly decreasing chains. So for $\delta>0$, $K$ does not exist for nonstandard parts and the requirement can only be satisfied for $\delta=0$ that is $K\in \mathbb{N}^{+0}\rightarrow \omega+1$. For the classic traces of section \ref{sect:standard:trace:semantics} where $\delta=0$ and $\mathbb{N}^{+0}=\mathbb{N}$, one can choose $K(n)=n$ (which is $t$ in (\ref{eq:Knuth:rule})) so that $\bigcup_{n\in\mathbb{N}^{+0}}K(n)\leqslant\omega$.

\subsection{Preventing Nonstandard Infinite Backward or Backward/Forward traces}
The idea is that traversing standard traces backward should terminate, which is the inverse of the requirement (\ref{eq:no:nonstandard:forward:traces}) that forward traversal of forward traces should terminate when program termination is required (that is $\bigcup_{n\in\mathbb{N}^{+0}}K(n)<\omega$). It follows that we can exclude transfinite nonstandard backward traces by an additional proof obligation applying the inverse of Knuth's  proof method. 
\begin{eqntabular}{l}
\exists K\in \mathbb{N}^{-\delta}\rightarrow\mathbb{O}\ \textrm{decreasing (antitone)}\ \mathrel{.}\forall n\in\mathbb{N}^{-\delta}\mathrel{.}\forall\alpha\in\mathbb{O}\label{eq:no:nonstandard:backward:traces}\\
\quad\{P\wedge\texttt{B}\wedge K(n)=\alpha\oplus 1\}\,\texttt{Sb}\,\{P\wedge K(n+1)=\alpha\}
\wedge{}\nonumber\\
\qquad (P\wedge K(n)=0)\Rightarrow\neg\texttt{B}\wedge\bigcup_{n\in\mathbb{N}^{-0}}K(n)<\omega\nonumber
\end{eqntabular}
For example, for a backward nonstandard trace with domain $\mathbb{N}^{-1}$, we would have
\begin{eqntabular*}[fl]{@{\qquad}c}
\underbrace{\ldots\ \begin{array}[t]{c}\begin{array}[b]{c}\texttt{B}\\-2_{-1}\end{array}\\\vphantom{2}\end{array}\ \begin{array}[b]{c}\texttt{B}\\-1_{-1}\end{array}\ \begin{array}[b]{c}\texttt{B}\\0_{-1}\end{array}\ \begin{array}[b]{c}\texttt{B}\\1_{-1}\end{array}\ \begin{array}[b]{c}\texttt{B}\\2_{-1}\end{array}\ \ldots}_{K\textrm{ \scriptsize not definable for\ }\delta>0}\ \underbrace{\begin{array}[t]{c}\ldots\\\ldots\end{array}\mbox{$\begin{array}[t]{c}\begin{array}[b]{c}\texttt{B}\\-2_0\end{array}\\2\end{array}$} \mbox{$\begin{array}[t]{c}\begin{array}[b]{c}\texttt{B}\\-1_0\end{array}\\1\end{array}$}\ \begin{array}[t]{c}\begin{array}[b]{c}\neg\texttt{B}\\0_0\end{array}\rlap{\qquad$\leftarrow\ n\in\mathbb{N}^{-1}$}\\0\rlap{\qquad$\leftarrow\ K(n)\in\mathbb{O}$}\end{array}}_{K \textrm{ \scriptsize not bounded}}
\end{eqntabular*}
so that (\ref{eq:no:nonstandard:backward:traces}) does not hold. Requirement (\ref{eq:no:nonstandard:backward:traces}) will also exclude nonstandard infinite backward/forward traces. For the standard trace semantics there is no infinite backward trace so the requirement (\ref{eq:no:nonstandard:backward:traces}) only applies to finite traces, in which case, for deterministic programs, $K(n)$ is the number of remaining steps after $n$ previous execution steps. For unbounded nondeterminism, the bound $\omega$ must be replaced by a larger ordinal depending on the cardinality of $\Sigma$ \cite[section 10, p.\@ 292]{CousotCousot85-AMS}

\section{Standard Hoare Logics}\label{sec:StandardHoareLogics}
We have enriched the axiomatic semantics with two additional proof obligations eliminating nonstandard infinite traces. Should these additional requirements be included in proofs by Hoare logic? The answer is no provided Hoare logic is explicitly specified to be interpreted in the standard trace semantics. In that case the additional verification conditions of section \ref{sect:logic:standardisation} are always satisfied. This hypothesis is usually not specified in the definition of axiomatic semantics by Hoare logic, and so almost always remains completely implicit. This is the case when Hoare logic and its variants are designed by abstraction
of the standard trace semantics
\cite{DBLP:journals/pacmpl/Cousot24,DBLP:journals/pacmpl/VerschtK25}.

But there are exceptions. For example, in Event-B  the traces are explicitly specified to be the standard ones \cite[Ch.\ 14.3]{Abrial:Event-B} and machines are specified by defining the states and invariants satisfying proof obligations which are essentially Hoare proof rules \cite{Cousot-Abrial-2026}. So specifications in Event-B, although based on an axiomatic semantics,  cannot have nonstandard interpretations.

\section{Static Analysis by Abstraction of the Axiomatic Semantics}
The axiomatic semantics has been used to prove the soundness of static analyses (e.g.\ \cite{DBLP:conf/cpp/Appel11}) or typing algorithms (e.g.\ \cite{DBLP:journals/jfp/NanevskiMB08,DBLP:conf/csfw/BeringerH07,DBLP:conf/csfw/BartheDR04}). However, it is not very well adapted for static analyzers based on program transformations for which Hoare logic extensions must be used \cite{DBLP:conf/popl/Gerhart75,DBLP:conf/popl/Benton04}. For example Astrée's loop unrolling \cite{DBLP:conf/pldi/BlanchetCCFMMMR03} is better explained by Burstall's proof method \cite{DBLP:conf/ifip/Burstall74,Cousot-RR-LRIM-83-04-sep-1983} or by cyclic proofs in Hoare logic \cite{DBLP:journals/corr/abs-2504-14283}. Of course all these logics have nonstandard interpretations that can be excluded by the standardisation methods proposed in section \ref{sect:logic:standardisation}.

\section{Abstract Hoare logics}
Abstract Hoare logics have been introduced in \cite{DBLP:conf/oopsla/CousotCLB12} with the idea that state properties in $\wp(\mathbb{E})$ in Hoare triples, can be replaced by abstract properties chosen in an abstract domain $\pair{L}{\sqsubseteq}$ whose meaning is given by a Galois connection $\pair{\wp(\mathbb{E})}{\subseteq}
\galois{\alpha}{\gamma}\pair{L}{\sqsubseteq}$ \cite[Ch.\ 11]{DBLP:books/mit/C2021}. If $\gamma$ is increasing but not join-preserving then the abstract Hoare rules may be sound while the intersection rule is incorrect \cite[p.\ 219]{DBLP:conf/oopsla/CousotCLB12}. Even though the concrete Hoare logic may be standard, the abstraction may introduce nonstandard executions. The standardisation methods of section \ref{sect:logic:standardisation} may not be applicable when the abstract domain $\pair{L}{\sqsubseteq}$ is not expressive enough.

\section{Related work}Another nonstandard transfinite trace semantics in $\wp(\mathbb{O}\rightarrow\Sigma)$ where $\mathbb{O}$ is the class of ordinals
suggested in \cite[footnote 17, page 53]{DBLP:journals/tcs/Cousot02} and fully developed in \cite{DBLP:journals/lisp/GiacobazziM03} yields another nonstandard semantic interpretation of axiomatic semantics (it is different from the nonstandard semantics of section \ref{sec:nonstandard:Traces} since $\mathbb{O}$ is not a Skolem model of Peano's naturals in which any $n\neq 0$ must have a predecessor, which is not the case for limit ordinals). Such transfinite forward traces can also be eliminated by Knuth's ordinal method of section \ref{sec:no:nonstandard:forward:traces}.

\section{Conclusion and Future Work}
Hoare postulated the logic named after him and relied implicitly on a model of computation which defines possible computations of programs. When later proved sound and complete, the same implicit model of computation was used \cite{DBLP:journals/siamcomp/Cook78,DBLP:journals/siamcomp/Cook81}. The design of logic by abstract interpretation of a trace semantics \cite{DBLP:conf/birthday/Cousot26,DBLP:journals/pacmpl/Cousot24} makes the same hypothesis. So nonstandard computation models were never considered. 

A solution to avoid nonstandard models of computation is to make explicit the assumption that the axiomatic semantics should be interpreted by standard computation models only, as explicitly specified in Event-B. Another solution is to add additional requirements (\ref{eq:no:nonstandard:forward:traces}) and (\ref{eq:no:nonstandard:backward:traces}) to restrict the axiomatic semantics to the standard computation model. Remarkably these requirements (\ref{eq:no:nonstandard:forward:traces}) and (\ref{eq:no:nonstandard:backward:traces}) need not be checked in proofs by Hoare logic since they are always satisfied in the standard model.

We anticipate that reverse Hoare logics \cite{DBLP:conf/sefm/VriesK11,DBLP:journals/pacmpl/OHearn20} can also be given nonstandard interpretations and restricted to the standard interpretation using Knuth's ordinal methods (\ref{eq:no:nonstandard:forward:traces}) and (\ref{eq:no:nonstandard:backward:traces}).

\subsubsection*{Acknowledgements.}
I thank the organizers of \href{https://helmutfest.github.io}{Helmut Seidel Fest} (TU München, Garching) where I first presented this work on July 16, 2026 and the reviewers for their suggestions. I used Claude (Opus 5) to check for misprints and grammatical errors and search for the bibliography's DOIs.

\bibliographystyle{splncs04}
\bibliography{bib}

\end{document}